\documentclass{article}

\usepackage{subcaption}
\usepackage{graphicx}

\usepackage{array}
\usepackage{rotating}
\usepackage{amsmath}
\usepackage{amsthm} 
\usepackage{amsfonts}
\usepackage{multirow}

\usepackage{hyperref}
\usepackage{caption}
\usepackage{subcaption}

\theoremstyle{plain}
\newtheorem{theorem}{Theorem}[section]

\theoremstyle{definition}
\newtheorem{definition}{Definition}[section]

\theoremstyle{remark}

\begin{document}

\title{HiDi: An efficient reverse engineering schema for large scale dynamic regulatory network reconstruction using adaptive differentiation}
\author{Yue Deng$^{1,2}$, Hector Zenil$^{1,2}$, Jesper Tegn\'er$^{2,3}$ and Narsis A. Kiani$^{1,2}$\thanks{N.A.K. is co-first author. Corresponding: hector.zenil@ki.se and narsis.kiani@ki.se}\bigskip \\
$^{1}$ Information Dynamics Lab, $^{2}$ Unit of Computational Medicine,\\ Center for Molecular Medicine, Department of Medicine,\\ Solna and Science for Life Laboratory\\ (SciLifeLab), Karolinska Institute, Stockholm, Sweden\\
$^3$ Biological  and  Environmental  Sciences  and  Engineering  Division,\\ Computer,  Electrical  and  Mathematical  Sciences  and  Engineering\\  Division, King  Abdullah  University  of  Science  and  Technology\\(KAUST), Thuwal,  Kingdom  of  Saudi  Arabia}
%$^*$ Co-first author\\
%$^+$ Corresponding author

\maketitle

\begin{abstract}

\noindent{\textbf{Motivation:} The use of differential equations (ODE) is one of the most promising approaches to network inference. The success of ODE-based approaches has, however, been limited, due to the difficulty in estimating parameters and by their lack of scalability. Here we introduce a novel method and pipeline to reverse engineer gene regulatory networks from gene expression of time series and perturbation data based upon an improvement on the calculation scheme of the derivatives and a pre-filtration step to reduce the number of possible links. The method introduces a linear differential equation model with adaptive numerical differentiation that is scalable to extremely large regulatory networks.\\
\textbf{Results:} We demonstrate the ability of this method to outperform current state-of-the-art methods applied to experimental and synthetic data using test data from the DREAM4 and DREAM5 challenges. Our method displays greater accuracy and scalability. We benchmark the performance of the pipeline with respect to data set size and levels of noise. We show that the computation time is linear over various network sizes. \\
\textbf{Availability:} The Matlab code of the HiDi implementation is available at: www.complexitycalculator.com/HiDiScript.zip \\
\textbf{Supplementary information:} Supplementary data are available at \textit{Bioinformatics}
online.}

\end{abstract}

\section{Introduction}

In recent years, we have witnessed remarkable advances in measurement technologies associated with data-acquisition techniques in many areas, including biological domains and engineering. We can simultaneously measure the expression levels of thousands of genes under different experimental conditions. Techniques to extract useful information from big data are needed. A fundamental reason for the difficulty in attempting to analyze many large scale data sets is low measurement density~\cite{siegenthaler}. Such datasets are high-dimensional and consist of a small number of points in a high-dimensional space, and it can be extremely challenging to find structures in such data. 

Networks provide a fundamental setting for representing and interpreting information. Network inference deals with the reconstruction of molecular networks directly from experimental data. Data is becoming ever larger, more dynamic, heterogeneous, noisy and incomplete, exacerbating the difficulty of the network inference challenge. For example, in the construction of genetic regulatory networks, we can simultaneously measure the expression of thousands of genes, but the number of experiments we can perform is limited when set against the combinatorial explosion of possible testable conditions. 
What sometimes makes an approximate inference possible in such situations is that many such systems in the real world have characteristic properties that allow them to be represented or approximated with a much smaller number of parameters than the dimensions of the phase space. 
On the one hand, there is a wide variety of different approaches available that can be used to infer molecular networks from experimental data, but a reliable network inference method remains a challenge and an open area of intensive research. Mathematical gene regulation network (GRN) models range from logical models with only Boolean values to continuous ones including detailed biochemical interactions (\cite{LIS,nar1,GENIE2010,TIGRESS2012}). 
On the other hand, logical models require fewer biological details and are computationally more efficient but also display limited dynamic behavior. In contrast, concrete models can describe more details of specific network dynamics but the computational cost of determining parameters becomes intractable. 

\begin{figure*}[!htbp]
\centering
 \includegraphics[width=1\linewidth]{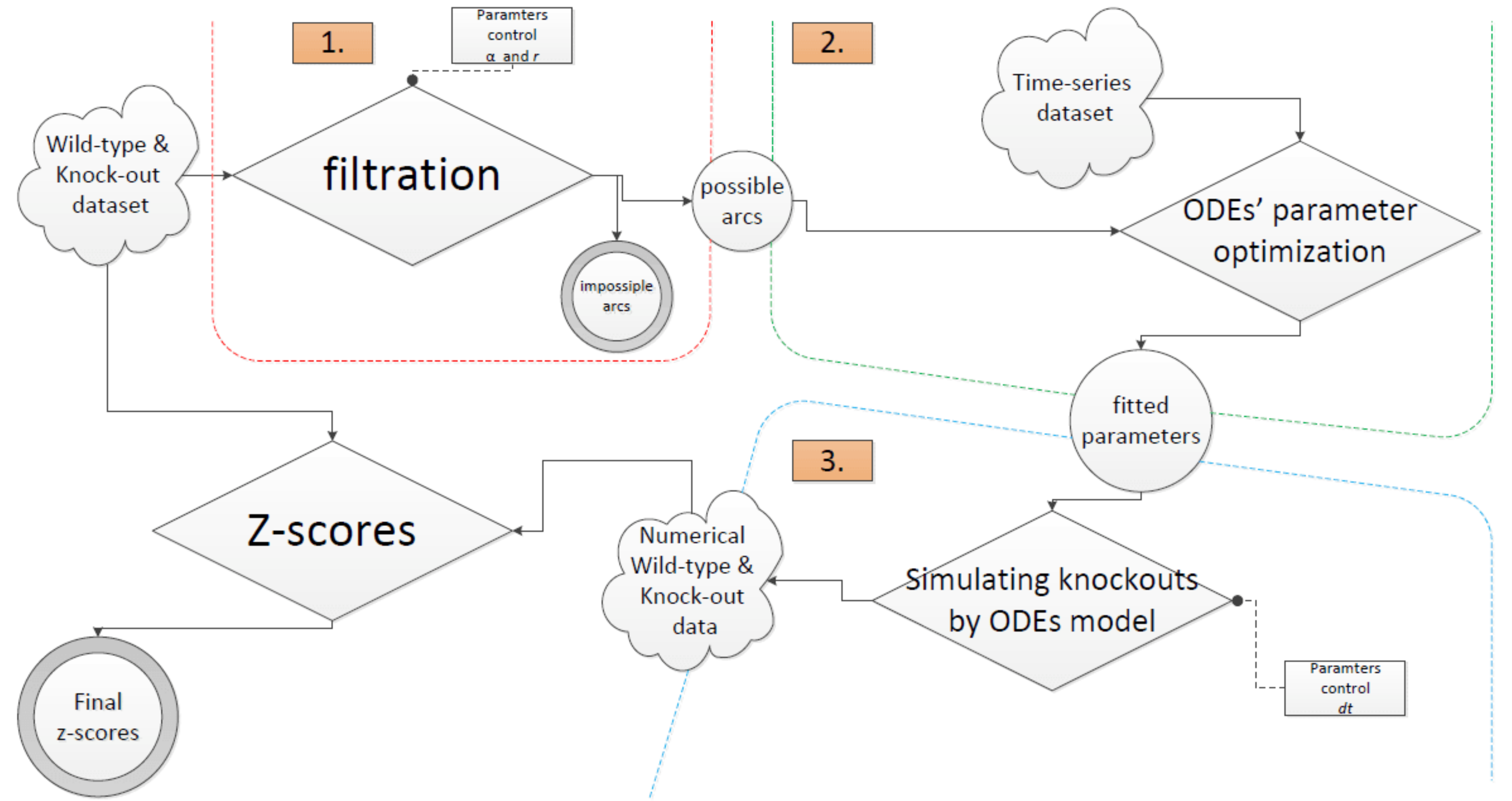}
  \caption{HiDi: Network inference pipeline including  pre-filtration, parameter identification and network inference.}
  \label{pipeline_graph}
\end{figure*}
In brief, given $N$ molecular elements, the challenge of network reconstruction is to decipher the complex interplay of the interacting molecules among an exponential number of possible topologies (all element interactions). For example, Median-Corrected Z-Scores~\cite{Greenfield2010}, Context Likelihood of Relatedness (CLR)~\cite{Aviv2010}, etc. can be applied to extract information about network topology from steady-state data. Methods based on steady-state data have the inherent weakness of being unable to readily distinguish between direct interactions and indirect interactions, because by the time the steady state is established the initial perturbation will have spread into the network. 
One of the most promising approaches for addressing this problem is the use of differential equations~\cite{Madar2009, Yip2010, James2014}. Here we interpret the problem of high-dimensional inference within the framework of linear differential equations. The main goal is to develop a flexible and principled pipeline for uncovering hidden structure underlying high-dimensional, complex data. Here, Ordinary Differential Equation (ODE) models are used to represent the dependence of the concentration of one molecular species on other molecular species. The linear ODE model~\cite{Madar2009}, nonlinear ODE model~\cite{Yip2010} and nonparametric additive ODE model~\cite{James2014} have been developed to cope with time-series (dynamic) data. These methods have the ability to detect transient perturbations in a network, but they require a large number of parameters to be determined. There are also other methods based on machine learning~\cite{raey}, singular value decomposition (SVD)~\cite{Yeung30042002}, Bayesian networks~\cite{10.1371/journal.pone.0029165}, etc., all with different limitations and degrees of success.

 In this paper we introduce a novel method with which to reverse-engineer gene regulatory networks from time-course perturbation data with greater accuracy and scalability than has been possible using other methods. Our method uses a linear differential equation model with adaptive numerical differentiation to identify extremely large regulatory networks. The method is computationally very cheap and the size of the network is no longer of paramount concern, both with respect to computational cost and data limitation. 
Furthermore, we address another challenge, the huge number of possible topologies, by embedding a filtration step within the method to reduce the number of free parameters before simulating the dynamical behavior. The latter is used to produce more information about the network's structure. We evaluate our method on simulated data, and study its performance with respect to data set size and levels of noise on a 1565-gene E.coli gene regulatory network. We show the computation time over various network sizes and estimate the order of computational complexity. We then use our method to study 5 networks in the benchmark collection DREAM4 Challenge. We show that our method outperforms current state-of-the-art methods applied to synthetic data and allows the reconstruction of bio-physically accurate dynamic models from noisy data, the only (and important) caveat being that it requires a large number of perturbation experiments.

\section{Methods}

\subsection{Differential Equations Model}

We represent a \textit{genetic regulatory network} (GRN) as a directed graph $G = (V, E)$, with nodes $V = \{v_1, \ldots, v_n\}$ corresponding to genes and directed edges $E = \{e_1, \ldots, e_m\}$ corresponding to regulatory interactions. An edge from vertex $i$ to vertex $j$ indicates that the expression of gene $i$, $x_i$, influences the expression of gene $j$, either by activating or by inhibiting it. We use a linear ODE model to model the dynamics of a GRN, that is, we assume that different regulators act independently, such that the total effect on the expression of gene $i$ can be written as the sum of the individual effects. This is clearly an oversimplification, and can be generalized by considering the products of the effects of different genes.

The ODE model can be written as:
\begin{equation} \label{ODE_Net2}
\frac{d\mathbf x}{dt} = \mathbf{a_0}+\mathbf{Ax}
\end{equation}
\noindent where $\mathbf {x} \in \mathbb {R}^{n\times 1},\mathbf{a_0}\in \mathbb R^{n\times 1},\mathbf A\in \mathbb R^{n\times n} $.\\
The \textit{basal synthesis} and \textit{degradation rate} for each gene is represented by $a_{ii}$. $X(t)$ is the concentration of genes at time t. $a_{ij}$ denotes the regulation strength of component $x_i$ on $x_j$.
$a_{ij}> 0$ corresponds to an activation, $a_{ij} < 0$ to an inhibition, and $a_{ij} = 0$ means that there is no regulation of gene $j$ by gene $i$. 

Even though using a linear ODE model is an oversimplification in a first approximation, its advantage is that the number of parameters is smaller as compared to non-linear models, and therefore the parameter values can be more accurately estimated, leading to more reliable descriptions of the dynamics of the gene regulatory networks, as we will show. This allows the method to scale up to a very large system. 

Traditionally, experimental data are organized in a matrix where the effect of perturbation of one gene in time  $\mathbf{x}^{t_i} \in \mathbb R^n,i=1,...,T $, and it can be encoded as follows:

\begin{equation*}
\mathbf X = \left[ \begin{array}{c} \mathbf x^{t_1} \\ \vdots \\ \mathbf x^{t_T} \end{array} \right] \in \mathbb R^{T\times n}
\end{equation*}
Applying different perturbations to the same gene leads to \textit{replicates}, of which the $r$-th replicate can be denoted as:

$$
\mathbf X_r \in \mathbb R ^{T\times n},
$$ 

Therefore, all $R$ replicates can be denoted in a series of matrices:

$$
\mathbf X_1,\mathbf X_2,...,\mathbf X_r,...,\mathbf X_R
$$

We employ the same notation for the observations of $\frac{d\mathbf x}{dt}$  in one perturbation experiment. Therefore:

\begin{equation*}
\mathbf Y = \left[ \begin{array}{c} \left.\frac{d\mathbf x}{dt}\right\vert_{t_1}\\ \vdots \\ \left.\frac{d\mathbf x}{dt}\right\vert_{t_T}\end{array} \right] \in \mathbb R^{T\times n};
\end{equation*}

as well as $R$ replicates:

$$
\mathbf Y_1,...,\mathbf Y_r,...,\mathbf Y_R.
$$\

Using these notations, we will reformulate the problem of network reconstruction as an optimization problem.

\subsection{Parameter estimation of ODE systems}

The estimation of model parameters from experimental time series data for differential equation models is typically carried out iteratively in two steps: (1) use of integration schemes to numerically generate solutions $x(t)$ for given parameters, and then (2) comparison of the model's prediction with the experimental data to calculate the error. 

\begin{figure}[!htbp]
\centering
\includegraphics[width=1\linewidth]{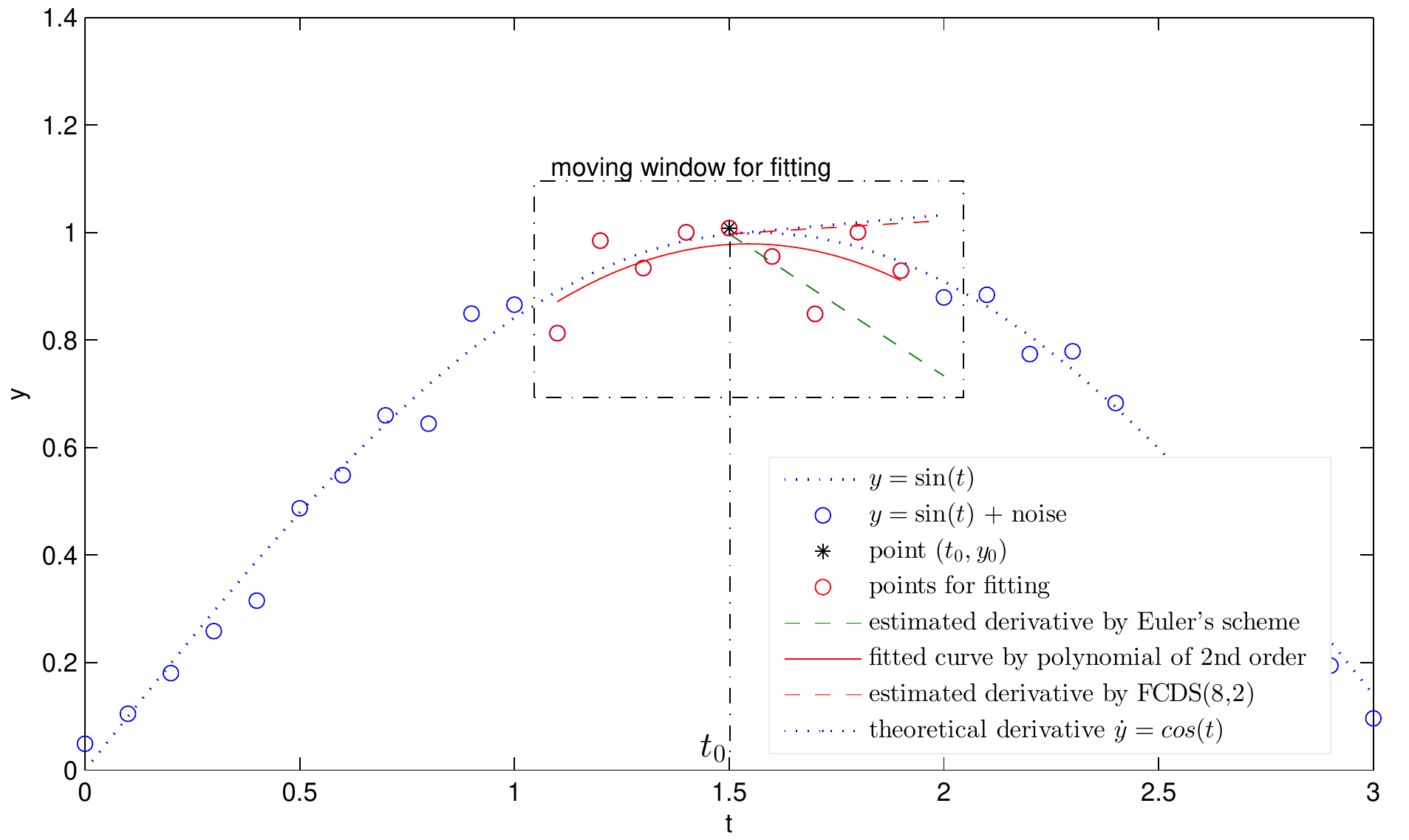}
\caption{An illustration of the derivative estimated by polynomial fitting using different schemes. The curve is $y=\sin(t)$, with theoretical derivative $\dot{y} = \cos(t)$ (blue dash line). The derivative at time point $t = t_0$ is estimated from the noisy data with $\sigma_0 = 0.1$ by fitting a 2nd order polynomial locally around 8 neighbors, i.e. by FCDS(8,2) (red dash line) and by Euler's scheme (green dash line).}
 \label{deri_ex}
\end{figure}

Initial values and model parameters are then modified to minimize this error. Prominent techniques, common in the literature, are the \textit{least square} (LS) and \textit{Kalman filtering methods}. When the number of parameters to be estimated increases, the \textit{loss of lock} problem grows very fast~\cite{strebel}. To circumvent this, here we convert the parameter identification problem in the ODE system into a minimization problem, but using \textit{Frobenius' norm}, denoted by $\|\mathbf . \|_F $.\\

The minimization problem can be written as: 

\begin{center}
\begin{equation}
J(\tilde{\mathbf A}) = \frac{1}{2R}\|\mathbf{D_y} - \mathbf{D_x} \tilde{\mathbf A}^T\|_F^2.
\end{equation}
\end{center}
where
$$ 
\begin{array}{ccc}
 \mathbf{D_y}=\left[\begin{array}{c}\mathbf Y_1\\ \vdots \\ \mathbf Y_r\\ \vdots\\ \mathbf Y_R\end{array}\right] & and &  \mathbf{D_x}=\left[\begin{array}{cc}\mathbf 1 &\mathbf X_1\\ \vdots & \vdots \\ \mathbf 1 & \mathbf X_r\\ \vdots & \vdots \\ \mathbf 1 & \mathbf X_R\end{array}\right]
\end{array}.
$$

From now on we refer to $\mathbf{D_y}$ and $\mathbf{D_x}$  as \emph{Derivative} matrix and \emph{Design} matrix respectively.

\begin{figure*}[htpb]
\begin{subfigure}{0.5\textwidth}
\includegraphics[scale=0.4]{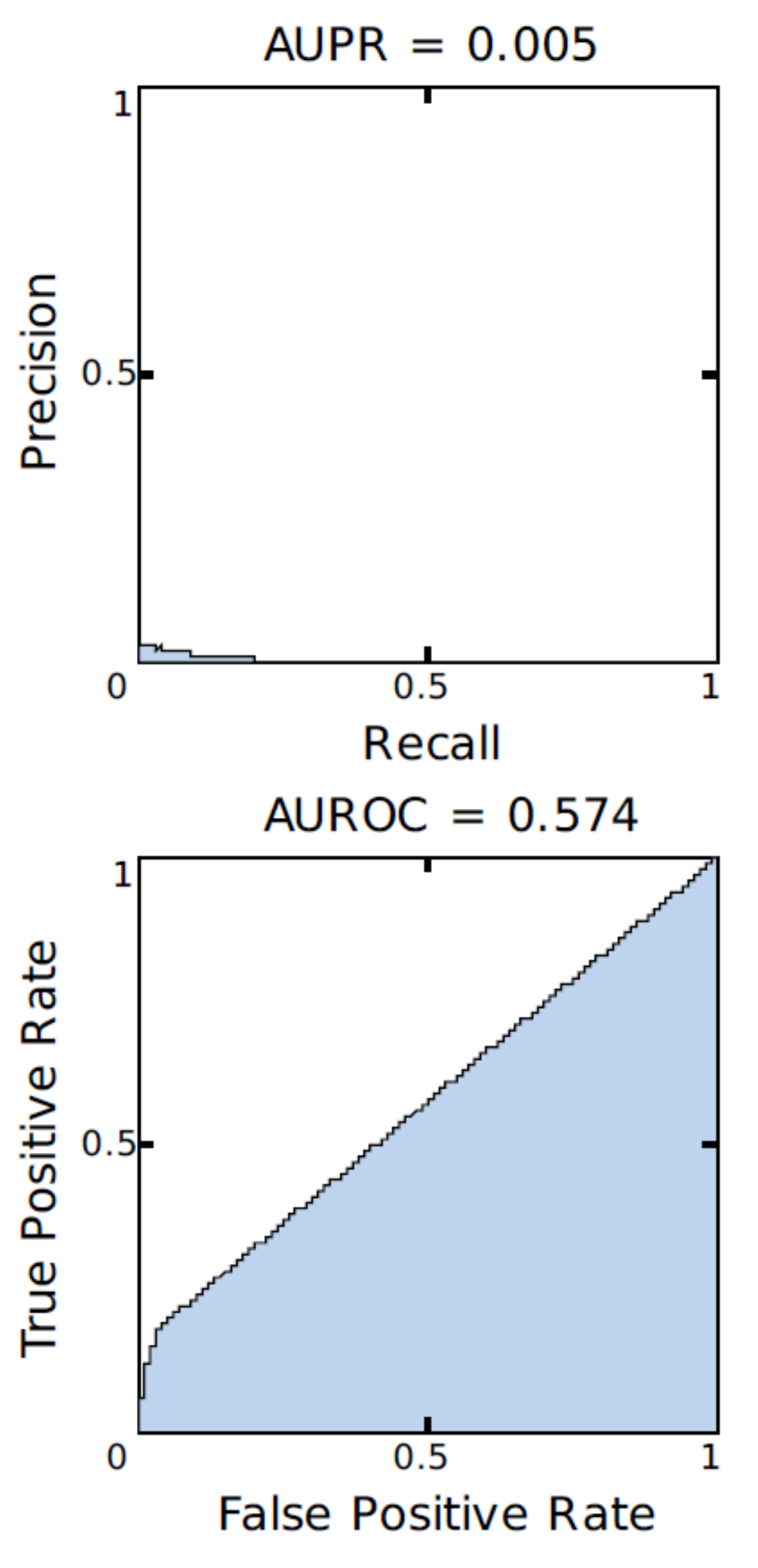} 
\caption{performance on data size = 300}
\label{fig:curvesizeb}
\end{subfigure}
\begin{subfigure}{0.5\textwidth}
\includegraphics[scale=0.4]{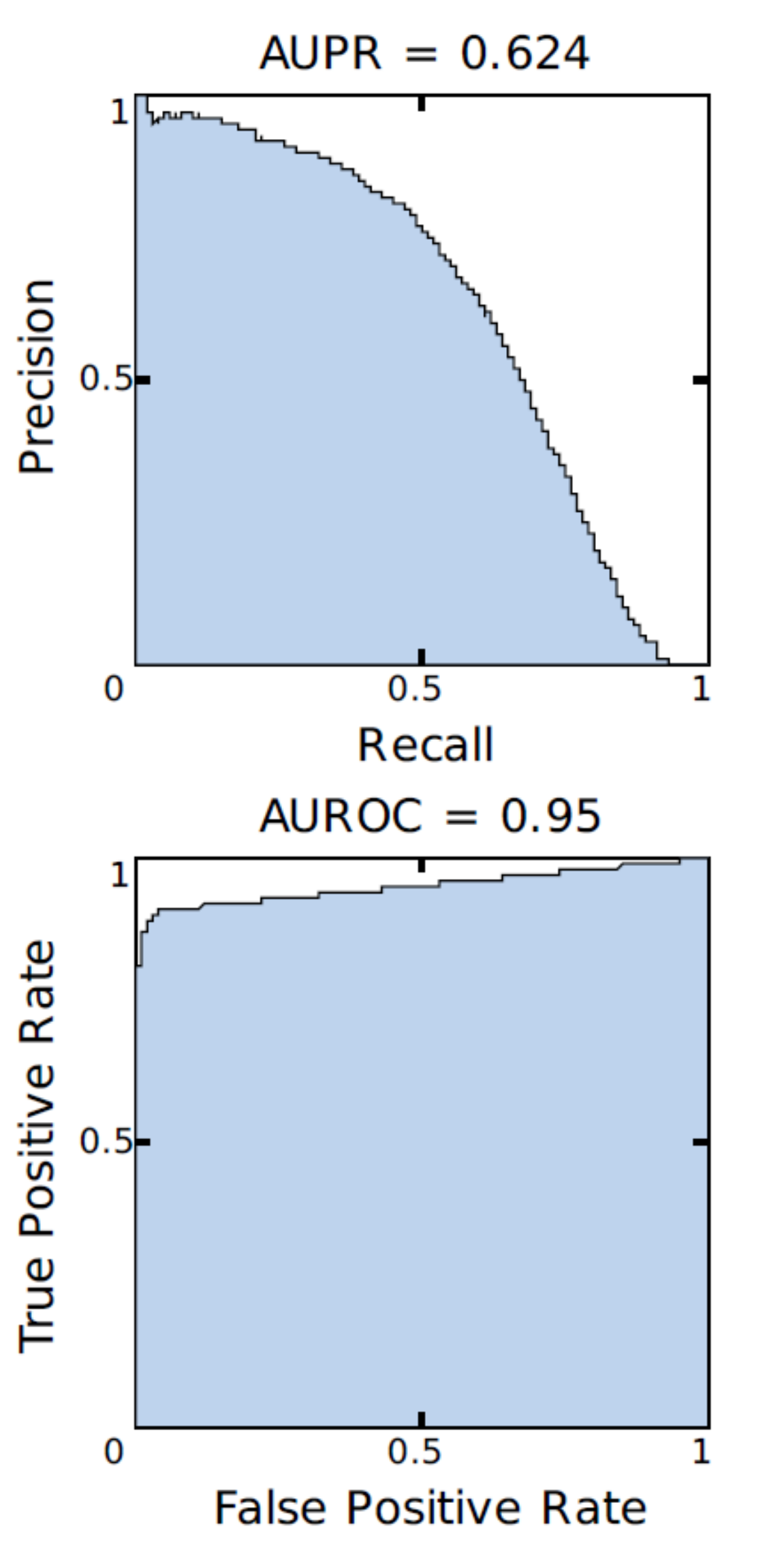}
\caption{performance on data size = 3\,600}
\label{fig:curvesizea}
\end{subfigure}
\caption{Precision-Recall Curve and ROC Curve of the prediction of ODE model with FCDS(8,8) scheme on data size 3\,600 and 300 respectively.}
\end{figure*}

In many cases, some knowledge of the biological processes underlying a particular data set will already be available. This can be used to formulate additional constraints to be included in the optimization problem or to add a regularization term to the objective function:

\begin{equation}
\begin{split}
\min_{\tilde{\mathbf A}\in \mathbb R^{n\times (n+1)}}   J(\tilde{\mathbf A})  &:= \frac{1}{2R}\|\mathbf{D_y} - \mathbf{D_x} \tilde{\mathbf A}^T\|_F^2 + \frac{\alpha}{2R}\|\mathbf A\|^2_F\\ \mbox{\ \ subject to}\\
a_{kl} &= 0,  \ \forall (k,l)\in \boldsymbol{C}
\end{split}
\label{P1}
\end{equation}

\noindent where $a_{kl}$, an entity of $\mathbf A$, and $\boldsymbol{C}\subset \{(i,j)|i,j\in \{1,..,n \}\}$ contain all constraints. $\alpha$ can be determined via cross validation. 

We use the \textit{Lagrange multipliers} $(\lambda)$ method to solve the above problem. We introduce a new variable $\lambda$ into the objective function. Thus we have:

\begin{equation} \label{Lag}
L(\tilde{\mathbf A},\mathbf{\lambda})  = \frac{1}{2R}\|\mathbf{D_y} - \mathbf{D_x} \tilde{\mathbf A}^T\|_F^2 + \frac{\alpha}{2R}\|\mathbf A\|^2_F  + \frac{1}{R}\sum_{(k,l)\in \boldsymbol{C}} \lambda_{kl}a_{kl}
\end{equation}
To find the global minimum, we should solve the following linear system.
\begin{equation}\label{linear}
\left\{\begin{array}{rcl}  \frac{\partial L(\tilde{\mathbf A},\mathbf{\lambda})}{\partial \tilde{\mathbf A}^T} &=& \mathbf 0  \\\\  \frac{\partial L(\tilde{\mathbf A},\mathbf{\lambda})}{\partial \lambda_{kl}} &=& 0 ,\ \forall (k,l)\in \boldsymbol{C} \end{array}\right..
\end{equation}
In order to solve this linear system, we use a \textit{vectorization operator}.
\begin{definition}[vectorization operator]\label{vec}
Let $\mathbf A = \mathbf{[a_1,...,a_i,...,a_n]}\in \mathbb R^{m\times n}$ and $\mathbf{a_i}\in \mathbb R^{m\times 1}$ be the i-th column of $\mathbf A$; the  \textit{vectorization operator} $vec:\mathbb R^{m\times n} \to \mathbb R^{mn\times 1}$ maps $A$ into a column vector by queuing the column vectors of $\mathbf A$ to the rear of the queue one by one: $$ vec(\mathbf A) = \left[ \begin{array}{c} \mathbf{a_1} \\ \mathbf{a_2} \\ \vdots \\ \mathbf{a_n} \end{array}\right]\in \mathbb R^{mn\times 1}.
$$ \end{definition}
With the vectorization operator, the linear system (\ref{linear}) can be written in a matrix form:

$$
\left[\begin{array}{cc}
\mathbf P & \mathbf{E}_{\boldsymbol{C}} \\ 
\mathbf{E}^T_{\boldsymbol{C}} & \mathbf 0
\end{array}\right]
\left[\begin{array}{c}
vec(\tilde{\mathbf A}^T) \\ 
\boldsymbol{\lambda}
\end{array}\right]
=
\left[\begin{array}{c}
vec(\mathbf{D_x}^T\mathbf{D_y}) \\ 
\boldsymbol{0}
\end{array}\right]
$$
\noindent where
\begin{align*}
\mathbf P\ &=\mathbf I_{n+1}\otimes \left(\mathbf{D_x}^T\mathbf{D_x} + \alpha \hat{\mathbf E}\right)\in \mathbb R^{(n+1)^2\times (n+1)^2},\\
\mathbf{E}_{\boldsymbol{C}}&=\left[...,vec(\mathbf{E}_{kl}),...\right]\in \mathbb R^{(n+1)^2\times |\boldsymbol{C}|},\\
\boldsymbol{\lambda}\ &=[...,\lambda_{kl},...]^T\in \mathbb R^{|\boldsymbol{C}|\times 1},\\
\mathbf{E}_{kl} &= [e_{ij}]\in\mathbb R^{(n+1)\times n} \text{\ with\ } e_{ij} = \delta^k_i\delta^l_j, i,j = 1,...,n.
\end{align*}

\noindent in which $(k,l)\in \boldsymbol{C}$, $|\boldsymbol{C}|$ is the number of elements or cardinality of set $\boldsymbol{C}$ and $\delta_i^j$ is the Kronecker delta.

The coefficient matrix
$$\left[\begin{array}{cc}
\mathbf P & \mathbf{E}_{\boldsymbol{C}} \\ 
\mathbf{E}^T_{\boldsymbol{C}} & \mathbf 0
\end{array}\right]
$$ is called the \textit{Karush-Kuhn-Tucker (KKT)} matrix, and it is nonsingular if and only if  $\mathbf{P+ E_C E_C^T}$ is positive definite.

\begin{figure*}[h]
\begin{subfigure}{0.5\textwidth}
\includegraphics[scale=0.4]{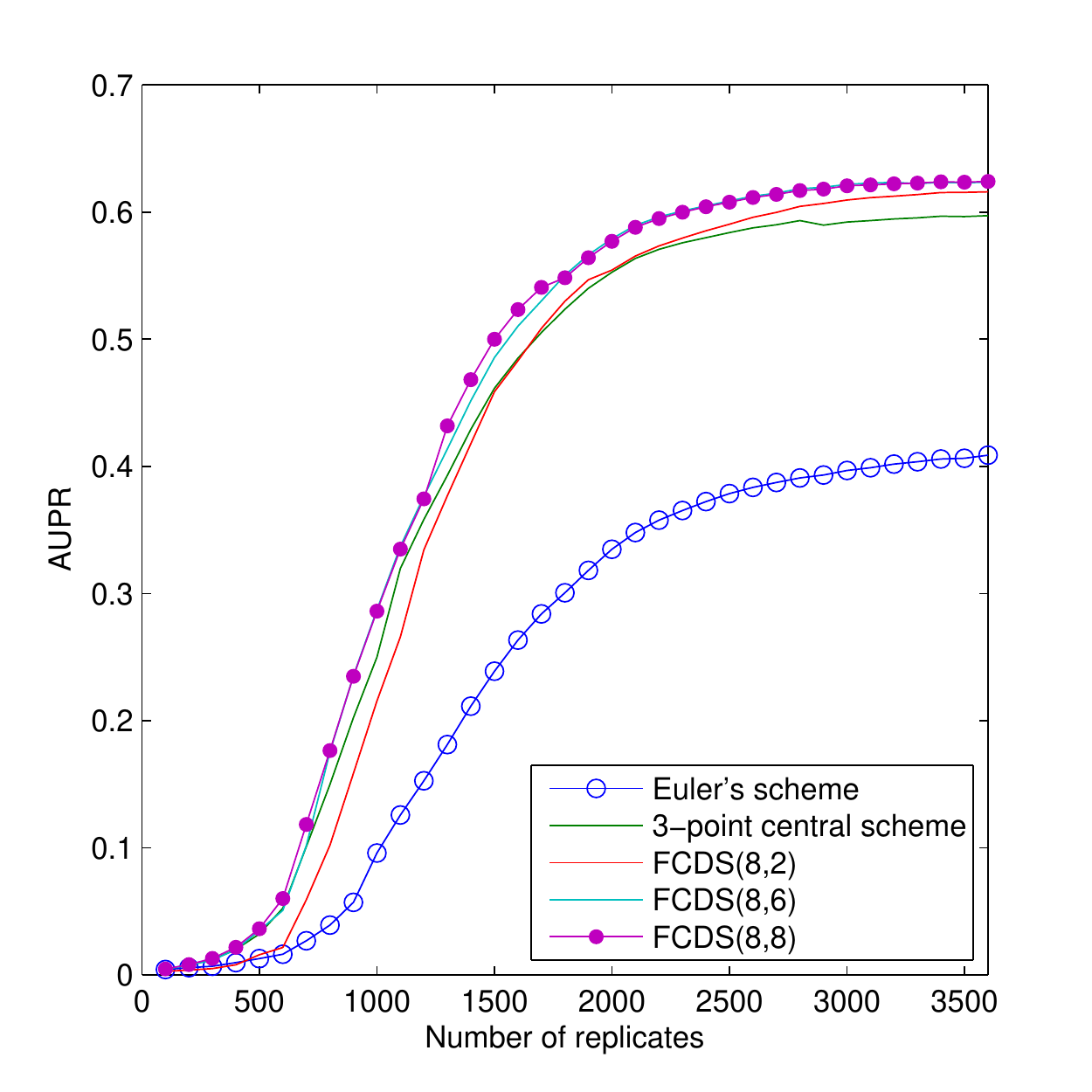} 
\caption{performance on data size = 300}
\end{subfigure}
\begin{subfigure}{0.5\textwidth}
\includegraphics[scale=0.4]{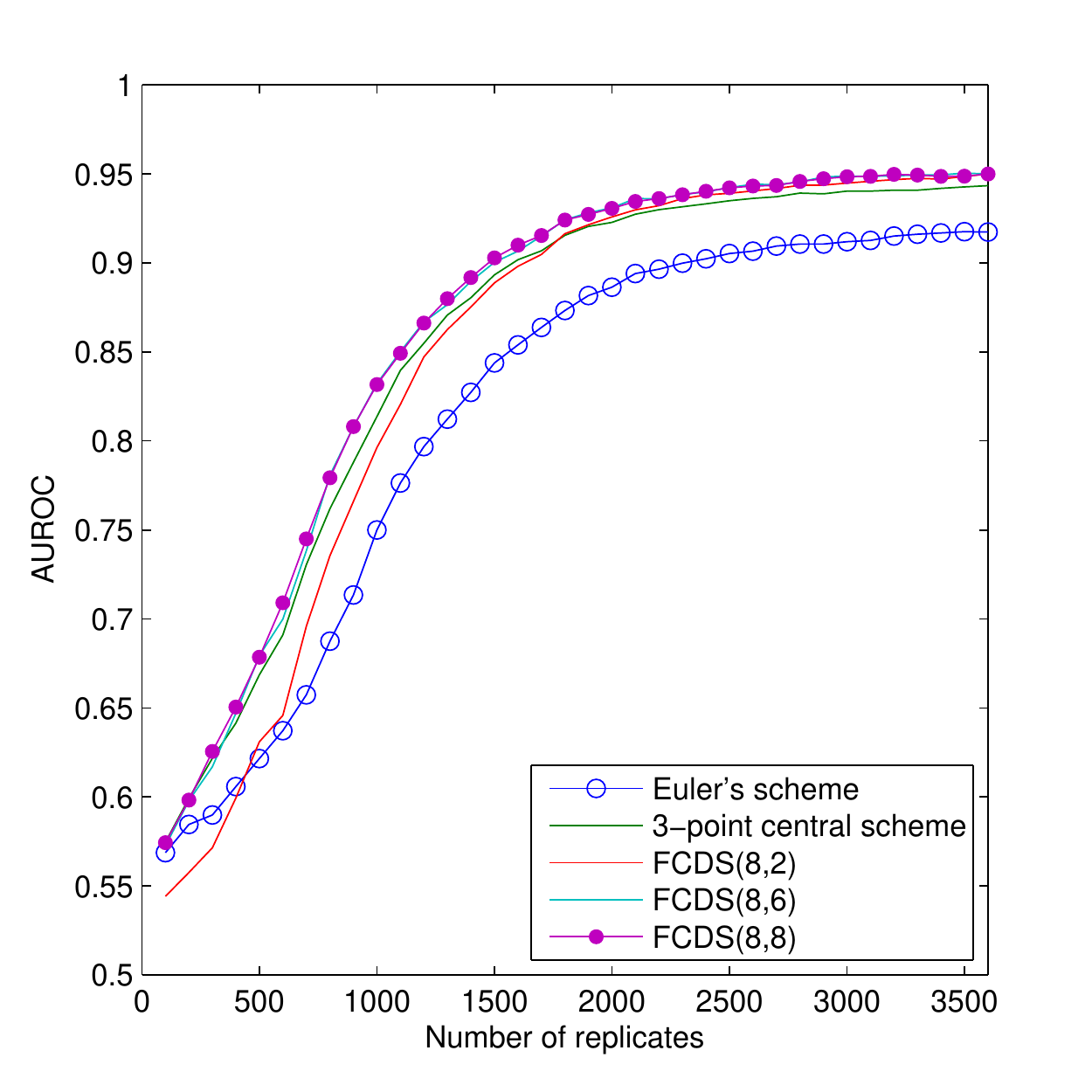}
\caption{performance on data size = 3\,600}
\end{subfigure}
\caption{Performance of ODE model without pipeline on different data sizes, with derivatives estimated by Euler's scheme, 3-point central scheme, FCDS(8,2), FCDS(8,6), and FCDS(8,8) schemes. Shown are AUROC and AUPR.}
\label{figdatasize}
\end{figure*}

\subsection{Numerical Differentiation of Noisy Data} \label{NoisyDiff}

Derivative matrices, in most cases, should be estimated from the observed $\mathbf x$ data which is  noisy due to measurement errors. Instead of the finite difference formulas such as the \textit{explicit Euler scheme} or the \textit{3-point central difference scheme} that are not a good choice for noisy data, we propose the use of polynomial fitting rather than interpolation, and only central difference schemes, which yield greater accuracy~\cite{deuflhard2002scientific}.

\textit{The main idea is to fit a polynomial locally with a few neighbor points and then differentiate the fitted polynomial.}

\begin{theorem}[Fitted Central  Derivative Scheme] \label{FCDS}
Let $P_n(t)$ be a polynomial of order $n$:
$$P_n(t) = a_0 + a_1(t-t_0) + ...+ a_n(t-t_0)^n,$$ 
fitted into $(t_0,x(t_0))$ and its $m=2k$ neighbor nodes:
$$
\begin{array}{lllllll} t_0-kh & ... & t_0-h & t_0 & t_0 + h & ...& t_0 + kh \\  x(t_0-kh) & ... & x(t_0-h) & x(t_0) & x(t_0 + h) & ...& x(t_0 + kh)\end{array}
$$ then the derivative of $x$ at the point $t = t_0$ can be approximated by $a_1$, which can be solved from the following linear system: \begin{equation} \label{solution}
(V^TV + \lambda I)\left[ \begin{array}{c}a_0\\a_1\\\vdots\\a_n\end{array}\right] = V^T\left[ \begin{array}{c}x(t_0-kh)\\ \vdots\\x(t_0)\\ \vdots\\x(t_0+kh)\end{array}\right]
\end{equation}
where $V$ is the Vandermonde matrix
$$
V = \left[ \begin{array}{ccccc} \vdots& \vdots &\vdots &\vdots& \vdots \\ 1 & -ih & (-ih)^2 & \hdots & (-ih)^n \\  \vdots& \vdots &\vdots &\vdots& \vdots \end{array}\right]\in \mathbb R^{(m+1)\times(n+1)},
$$ $ i = -k,...,0,...,k$
\end{theorem}
\begin{proof}
The parameters in the polynomial can be fitted by the classical linear least squares regression:
$$
\min_{a_0,a_1,...,a_n\in\mathbb R} \frac{1}{2}\sum_{i=-k}^k|x(t-ik)-P_n(t-ih)|^2 + \frac{\lambda}{2} \sum_{i=0}^na_i^2
$$, 
of which the solution is well known, as shown in equation (\ref{solution}) in this theorem. 
Furthermore, the estimation of $dx/dt$ is:
$$
\left.\frac{d\mathbf x}{dt}\right\vert_{t = t_0} \approx \left.\frac{P_n(t)}{dt}\right\vert_{t=t_0} = a_1
$$
\end{proof}

We call this scheme the \textit{Fitted Central Derivative Scheme} and it is denoted $FCDS(m,n)$. $m$ is the window length of the moving fitting, indicating how many nodes are involved; $n$ is the order of the fitted polynomial and shows the accuracy of the difference scheme. Once $m$ and $n$ are given, the scheme can be determined. If $m\geq n$, the solution is unique. As $n$ decreases, the fitted polynomial becomes smoother and the accuracy of differentiation decreases. When $n=m$, the schemes are the classical central difference schemes, which, however, are the worst at noise control because they are deduced from Lagrange interpolation under the assumption that the curve passes exactly through these points, which is not the case for noisy data. If $m = n$ and $\lambda=0$, then the polynomial fitting collapses to a Lagrange interpolation, which yields the classical $(m+1)$-point central difference scheme. Fig.\ref{deri_ex} is an example illustrating that FCDS is more robust than Euler's scheme.

%\subsection{Parameter Filtration of Sparse ODE System} \label{Sparsity}
\subsection{Network inference} \label{pipelineSec}In the previous section
 we suggested an approach for dealing with noisy data. The introduction of FCDS  is key to coping with experimental and biological variability. We can now solve the linear system (\ref{linear}) to find the network underlying the data. This would be reasonable if the problem were not undetermined and the data were not noisy. However, since the optimization problem encountered for real data is typically underdetermined, marginal distributions for parameters show wide peaks and large confidence intervals. In many biological settings, available data are insufficient to unambiguously reconstruct the underlying network. In such situations, strict regularization of the objective function, using, for example, \textit{maximum parsimony}, or the inclusion of additional prior biological knowledge can be helpful (\cite{nar1}). In this section we show how we can drive solutions to sparse networks, and make the inference feasible even in the presence of substantial amounts of missing data. We show how prior knowledge can be formulated as a constrained optimization problem. We use a filter for a sparse ODE system based on three outlier detection techniques: \textit{generalized (extreme studentized deviate), {ESD} test}~\cite{Rosner1983} and \textit{ modified Z-Scores}~\cite{Boris1993}. An alternative is to use normal differential gene expression analysis~\cite{markowetz2007} (not very different from our approach) to avoid the need to fine-tune one of the parameters ($r$). However, this may detect too many or too few genes as outliers and thus reintroduce a parameter. We use a filtration step and approximate the derivative matrix numerically and simulate the dynamics of the model with parameters for each knockout, and then compare the simulation results at each step. The workflow of the proposed method is shown in Fig.~\ref{pipeline_graph}. The outlier detection first filters out unlikely links and thus some parameters in the ODEs are restricted to zero. The constrained optimization solution for identification of parameters described can then be applied. Thereafter, the knockout experiments can be simulated numerically. Unlike the steady-state knockout data, the ODE model can simulate the effects of knockout in an arbitrarily short time period. With the simulated transient knock-out data, Z-scores can be computed and normalized as final scores ranking the possible links. The time-series data $\mathbf X_r$ of gene expression level can be directly used for parameter optimization. The absolute values of the parameters were used as scores ranking confidence in the prediction of links: a larger $a_{ij}$ indicated a stronger influence of gene-$i$ on gene-$j$. The sign of the parameters tells whether the interaction is an inhibition or an activation. We filter out unlikely edges in the network by analyzing wild-type data and knock-out data. This step shrinks the parameters searching space, thus reducing the data size required for fitting.
 If gene-$i$ has an edge (or a path) leading to gene-$j$ and is knocked out in the experiment, the knock-out steady-state gene expression data of gene-$j$, denoted as $x_j^{ko_i}$, would be expected to change significantly from its wild-type level $x_j^{wt}$. So the Generalized ESD test would be able to detect the deviation $x_j^{ko_i}-x_j^{wt}$ as an outlier, and the modified Z-score will also be high. 

If gene-$i$ does not have many direct or indirect interactions with gene-$j$ and is knocked out in the experiment, it would be expected that  $x_j^{ko_i}$ will not change significantly, excepting noisy fluctuation from its wild-type level $x_j^{wt}$, so that the Generalized ESD test will not be able to detect the deviation $x_j^{ko_i}-x_j^{wt}$ as an outlier and the modified Z-score will not be high. 
There are two parameters to be specified in the filtration process: the upper bound of number of outliers $r$ and significance level $\alpha$. 
In the context of a network, the upper bound of number of outliers $r$ is the limit of in-degrees of all gene nodes. Since only up to $r$ outliers will be detected, only up to $r$ regulators per gene can be found. One can choose $r$ based on the size of the network to be inferred, prior knowledge or subjective expectations. One should also consider the number of data sets available, since a larger $r$ yields more non-zero parameters in ODE, which then will require more data for fitting. 
Significance level $\alpha$ ({type I error}) and $\beta$ (type II error) can be interpreted as:

$\alpha $ : P (an edge is detected when there is no edge)

$\beta$ : P (an edge is not detected but there is  an edge)

We need to choose a \textbf{large} significance level $\alpha$ when the outlier detection technique is applied to filtration, whereas we need to choose a \textbf{small} significance level $\alpha$ when only the outlier detection technique is applied in finding possible interactions.

\section{Results}
\label{performance}

\subsection{Data simulation}
We evaluate the performance of HiDi and other methods on six benchmark datasets, each consisting of a compendium of gene expression data, a gold standard set of verified interactions which we ideally would like to reconstruct. Expression data are either simulated or experimentally measured under a wide range of perturbations.  Data simulation has been done using GeneNetWeaver (GNW) (\cite{LIS,EPFL-ARTICLE-166759}), an \textit{in silico} (numerical) simulator containing  dynamic models of gene regulatory networks of \textit{E.coli}~\cite{gama2008regulondb} and \textit{S.cerevisiae}~\cite{Balaji2006213}, including a thermodynamic model of transcriptional regulation, mRNA and protein dynamics, enabling it to generate gene expression data. Table~\ref{perf_data} summarizes the statistics pertaining to these networks. We simulated 3 different types of data: gene expression, time course gene expression and perturbation data at steady state.

%To prevent this table from exceeding the column margin, "DREAM4 Multifactorial Network" is renamed "D4 Multi-net".

\begin{table}[!htbp]
{\small
\centering
    \begin{tabular}{|c|c|c|c|c|}
    \hline
Network&	TF&	Genes &		Chips&	 Verified links\\\hline
E. coli Network from&1565&	1565 &		1565&	 3758\\\hline
D4 Multi-net* 1&100&	100 &		100&	 176\\\hline
D4 Multi-net* 2&	100&	100 &		100&	 249 \\\hline
 D4 Multi-net* 3&	100&	100 &		100&	 195\\\hline
	D4 Multi-net* 4	&	100&	100 &		100&	211\\\hline
	D4 Multi-net* 5&	100&	100 &		100&	193 	\\ \hline
\end{tabular}
\vspace{2mm}
%todo: add correct references to perf-six in text
\caption{\label{perf-six}Six Benchmark data sets used in our study. *'D4 Multi-net' stands for 'DREAM4 Multifactorial Network'.}
}
\end{table}

\subsection{Performance}

We compare HiDi to several other GRN inference methods. We use results from the DREAM challenge and compare inferred networks to known regulations to assess the number of true positives (TP, the number of known regulations among the top predictions), false positives (FP, the number of predicted regulations in the top K which are not known regulations), false negatives (FN, the number of known interactions which are not among the top predictions) and true negatives (TN, the number of pairs not among the top predictions which are not known regulations). We then compute precision (TP/(TP $+$ FP)), recall (TP/(TP $+$ FN)). We assess globally how these statistics vary with the number of top predictions by computing the area under the receiver operating characteristic (ROC) curve and the precision-recall curve (AUROC and AUPR respectively). We also compute a p-value for the AUROC and AUPR scores, based on all DREAM participants' predictions as used by the DREAM5 organizers to rank the teams. This involves randomly drawing edges from the teams' prediction lists and computing the probabilities of obtaining an equal or larger AUPR (or AUROC) by chance. Finally, we compute a score for our inference method by integrating the AUROC and AUPR p-values, as was suggested by the DREAM organizers.

\subsection{DREAM4 and DREAM5 Network Challenges} \label{dream4re}

The wild-type data set, knock-out data set and 10 replicates of the time-series data set have been provided for each network. Since only 10 replicates of the time-series data are provided, which is insufficient to determine 10\,100 parameters for a 100-node network, we used our pre-filtration procedure to dramatically reduce the searching dimension. The effectiveness of filtration on one of the networks \textit{network-1} in the DREAM4 challenge is shown in Fig. 1 in the supplementary document.%~\ref{figfil}.
The original 9\,900 possible arcs in a 100-gene network were reduced to a few hundreds, so that the required data for parameter estimation was dramatically reduced. 

We first display the effect of choosing the significance level $\alpha$ in the pre-filtration procedure. Then we compare the performance of the pipeline method being presented here with those of participating teams, using the data provided by the challenge organizers.

\begin{table}[!htbp]
{\small
\centering
    \begin{tabular}{c|cc}
    & AUPR                      & p-value of AUPR                    \\ \hline
    Net1        &   \textbf{ 0.630} (T395: 0.536) &    \textbf{1.60E-150} (T395: 1.23E-121 ) \\
    Net2        &    \textbf{0.448} (T296: 0.396) &   \textbf{ 8.31E-206} (T296: 1.80E-177)  \\
    Net3        &    \textbf{0.413} (T395: 0.390) &    \textbf{7.94E-101} (T395: 5.20E-95)  \\
    Net4        &    \textbf{0.491} (T271: 0.403) &    \textbf{6.41E-117} (T271: 2.93E-95)   \\
    Net5        &    0.251 (T532: 0.326) &    2.78E-56 (T532: 3.82E-74)  \\
  
    \end{tabular}
    %\caption{AUPR of pipeline on 5 evaluated networks.  
\\
\vspace{5mm}
\centering
    \begin{tabular}{c|cc}
     & AUROC                     & p-value of AUROC                           \\ \hline
    Net1        &    0.916 (T548: 0.917) &    2.94E-41 (T548: 1.92E-41)  \\
    Net2        &    \textbf{0.868} (T395: 0.801) &   \textbf{ 3.65E-64} (T395: 4.33E-45)  \\
    Net3       &    0.797 (T515: 0.844) &    5.27E-39 (T515: 2.84E-51)   \\
    Net4         &    \textbf{0.852} (T549: 0.848) &    \textbf{4.20E-45} (T549: 2.56E-44)    \\
    Net5        &    \textbf{0.803} (T548: 0.778) &    \textbf{2.36E-35} (T548: 1.82E-30) \\
    \end{tabular}
\\
\vspace{5mm}
   % \caption{AUROC of pipeline on 5 evaluated networks. \textit{TXXX} shows the top score for a team participating in the DREAM Challenge}\ \\
   \begin{tabular}[htbp]{|c|c|}\hline
    Overall AURR score  &    \textbf{125.345}  (Team395: 103.068)            \\
    Overall AUROC score   &    \textbf{44.250} (Team548: 40.962)      \\
    Overall score       & \textbf{84.798} (Team395: 71.589)                                 \\\hline
    \end{tabular}
    \\
\vspace{5mm}
\caption{Performance of the pipeline on 5 networks in the DREAM4 challenge. (\textit{TXXX} refers to the top score for a team participating in the DREAM Challenge.) The top teams in each subcategory are in brackets, together with their scores; the bold item indicates that the score has exceeded the top one. The performance reported for HiDi is achieved with these settings: $\alpha=0.9,\ r = 20$ in pre-filtration, FCDS(8,6) scheme in ODE model and $\ dt = 0.1$ in the post-modeling stage.}\label{perftop}
}
\end{table}

Fig.~\ref{figalpha} shows that both the precision ${{TP}}/{{(TP + FP)}}$ and the False Negative Rate ${{FN}}/{{(FN + TP)}}$ decreased as $\alpha$ increased in the prediction of \textit{Network 1}. A pre-filtration is good if the False Negative Rate is low, while a prediction is bad if the precision is low. This is consistent with the statement in Section~\ref{pipelineSec} that a large $\alpha$ is recommended in the pre-filtration while a small $\alpha$ is advisable for separate use.

Table~\ref{perf_data_alpha} shows the performance on the five networks for different values of $\alpha$.

 Now, we compare the performance with those of participating teams. 

Compared to other teams we ranked first, with an overall score of 71.589, and we also came first in the sub-ranking of AUPR, with a score of 103.068; Team 548 ranked first in the AUROC, with a score of 40.962. Table~\ref{perftop} shows the performance of the proposed method on the five 100-gene networks with respect to the top performers in the challenge.
	
\begin{table}[!htbp]
{\small
\centering
    \begin{tabular}{|l|c|ccc|}
    \hline
\multicolumn{2}{|c|}{} &$\alpha=0.01$&$\alpha=0.5$& $\alpha=0.9$ \\
     \hline
\multirow{5}{*}{\rotatebox{90}{AUPR}} & Net1 &	0.571 	&	0.627 	&	0.630 	 \\
&	Net2	&	0.430 	&	0.460 	&	0.448 	\\
&	Net3	&	0.316 	&	0.407 	&	0.413 	\\
&	Net4	&	0.389 	&	0.467 	&	0.491 	\\
&	Net5	&	0.173 	&	0.249 	&	0.251 	\\ 
\hline
\multirow{5}{*}{\rotatebox{90}{AUROC}}&	Net1	&	0.878 	&	0.894 	&	0.916 		\\
&	Net2	&	0.783	&	0.854	&	0.868	 \\
&	Net3	&	0.769 	&	0.813 	&	0.797 	 \\
&	Net4	&	0.807 	&	0.852 	&	0.852 	 \\
&	Net5	&	0.711 	&	0.773 	&	0.803 	 \\
\hline
\multicolumn{2}{|c|}{{Overall AUPR Score}}   & 106.878 &	125.042	&	125.345 
	\\
 \hline 
\multicolumn{2}{|c|}{{Overall AUROC Score}}  &	31.969 	&	42.378 	&	44.250 		\\ \hline
\multicolumn{2}{|c|}{Overall  Score} &	69.424 	&	83.710 	&	84.798 
\\ \hline
\end{tabular}
\vspace{2mm}
%to do: add proper refs to perf_data_alpha
\caption{\label{perf_data_alpha}Performance on the DREAM4 challenge networks for different  values of  $\alpha$. All values have been calculated using the scoring scripts provided in the Dream challenge.}
}
\end{table}

We then use HiDi without filtration, provided adequate data are available. The software \textit{GeneNetWeaver (GNW) 3.1.1 Beta} can be set to generate the time-series data for the five networks in the DREAM4 Challenge, with the same model and the same noise level as in the challenge. This allows us to produce as much data as we want and to show how the performance is improved with more data, as is evident from Table~\ref{perf_data}.
 The method, with filtration, has produced quite a good result, especially with such a limited amount of data. However, without filtration it can achieve an even better performance, despite the fact that the amount of data required would be large. One can also see how poor its performance would be if applied on limited data without filtration.
\subsection{Reconstruction of the \textit{E.coli} GRN} \label{reecoli}
We tested the impact of data size on performance with the 1\,565-node \textit{E.coli} network. Fig.~\ref{figdatasize} shows that increase of replicates produces S-shape curves of both AUPR and AUROC. It indicates that the fitting problem always requires adequate data; with more data, one can expect a better performance. However, the performance has an upper limitation despite the fact that surplus data are supplied. In this test, the upper limitations of AUPR and AUROC for the FCDS(8,8) scheme are 0.624 and 0.95 respectively.

Fig.~\ref{fig:curvesizea} shows the Precision-Recall Curve and Receiver-Operating Characteristic (ROC) Curve when this limitation is achieved; Fig.~\ref{fig:curvesizeb} shows that when there is not enough data, the prediction is no better than mere random guesswork. 

Euler's scheme, the 3-point central scheme, the FCDS(8,2), FCDS(8,6), and FCDS(8,8) schemes have truncation errors on the order of  $O(h)$, $O(h^2)$, $O(h^2)$, $O(h^6)$, and $O(h^8)$ respectively. Fig.~\ref{figdatasize} also shows that a higher order of truncation error leads to a better performance. As regards details, Euler's scheme $\left(O(h)\right)$ turned in the worst performance. FCDS(8,2) and the 3-point central scheme (both with $O(h^2)$) yielded a similar performance in some ranges, but FCDS(8,2) eventually surpassed the latter.
 
The performances of FCDS(8,6) and FCDS(8,8) are very similar, which indicates that a truncation error of $O(h^6)$ is sufficient for this test and that an increase in accuracy leads to no improvement.

To further demonstrate the applicability of our method to real biological data, we applied HiDi to the four networks from the DREAM 5 challenge, consisting of one derived by in–silico simulation and three obtained experimentally from three species. Due to limited knock-out data provided for Network 2 and Network 4, we could not run the model for these networks, only for network 1 and network 3. The result has been shown in table ~\ref{perf_d5}. The reported performance for HiDi is achieved with these settings: $\alpha=0.9,\ r = 20$ in pre-filtration, FCDS(8,2) scheme in the ODE model and $\ dt = 0.1$ in the post-modeling stage.

\begin{table}[!htbp]
{\footnotesize
\centering
    \begin{tabular}{|c|c|c|c|c|}
    \hline
Method & 	Net1AUROC & Net1AUPR & Net3AUROC & Net3AUPR	\\\hline
GENIE3  & 0.815	  &	0.291&	0.617&	0.093 \\\hline
Other 2 &	0.78  &	0.245 &	0.671 & 	0.119 \\\hline
HiDi &	0.792   &	0.272  & 	0.638   &	0.105  \\\hline
 \end{tabular}
\vspace{2mm}
%todo: add correct references to perf-six in text
\caption{\label{perf_d5}Comparison between HiDi and top performing methods in DREAM5 challenges.}
}
\end{table}

\subsection{Computational time complexity} \label{comptime}

The computation time was evaluated by varying the size of the ODE system from 500 to 10\,000  variables, which implied a calculation of over 0.25 million to 100 million parameters to be optimized.

 The time for reconstruction of the 500-gene network was 0.6s, and 1766s (about 30min) for the 10\,000-gene network, in which over 100 million parameters were optimized; the $(\log_{10}n , \log_{10}t)$ was fitted into a straight line which shows that the time complexity of this algorithm is about $O(1/10^{7.6}\cdot n^{2.7})$.

\begin{table}[!htbp]
{\small
%\begin{subtable}[h!]{0.45\textwidth}
\centering
\begin{tabular}{|l|c|cccc|}
   \hline
    \multicolumn{2}{|c|}{Data Size} &R=10& R=110& R=1110&R=2110 \\ 
    \hline
     \multirow{5}{*}{\rotatebox{90}{AUPR}} & Net1 & 0.168 	&	0.496 	&	0.605 	&	\textbf{0.612} \\
    & Net2 & 0.100 	&	0.271 	&	0.429 	&	\textbf{0.438} \\
    &	Net3	&	0.074 	&	0.352 	&	0.487 	&	\textbf{0.496} \\
&	Net4	&	0.124 	&	0.409 	&	0.540 	&	\textbf{0.553}  	\\
&	Net5	&	0.046 	&	0.293 	&	0.434 	&	\textbf{0.441}  	\\ \hline
     \multirow{5}{*}{\rotatebox{90}{AUROC}}&	Net1	&	0.774 	&	0.863 	&	0.925 & \textbf{0.935}	 	\\
&	Net2	&	0.654 	&	0.744 	&	0.870 	&	\textbf{0.881}	  \\
 &	Net3	&	0.641 	&	0.800 	&	0.851 	&	\textbf{0.864}	  \\
&	Net4	&	0.720 	&	0.820 	&	0.890 	&	\textbf{0.913}	  \\
&	Net5	&	0.644 	&	0.789 	&	0.871 	&	\textbf{0.881}	  \\ \hline
\multicolumn{2}{|l|}{{Overall AUPR}} & 22.270 	&	94.697 	&	136.954 	&	\textbf{139.813}   	\\ \hline 
\multicolumn{2}{|l|}{{Overall AUROC}} &	14.893 	&	34.335 	&	52.182 	&	\textbf{55.495}   	\\ \hline
\multicolumn{2}{|l|}{Overall Score} &	18.581 	&	64.516 	&	94.568 	&	\textbf{97.654}   	\\ \hline
    \end{tabular}\\

%\end{subtable}
\vspace{5mm}
%\begin{subtable}[h!]{0.45\textwidth}
\centering
\begin{tabular}{|l|c|cc|}
\hline
    \multicolumn{2}{|c|}{Data Size} &R*=10& R**=10 \\ 
    \hline
     \multirow{5}{*}{\rotatebox{90}{AUPR}} & Net1 & \textbf{0.630} 	&	0.536  \\
    & Net2 & \textbf{0.448} 	&	0.396 \\
    &	Net3	&	\textbf{0.413} 	&	0.390  \\
&	Net4	&	\textbf{0.491} 	&	0.403 \\
&	Net5	&	0.251 	&	\textbf{0.326} \\ \hline
     \multirow{5}{*}{\rotatebox{90}{AUROC}}&	Net1	&	0.916 	&	\textbf{0.917} \\
&	Net2	&	\textbf{0.868} 	&	0.801 		  \\
 &	Net3	&	0.797 	&	\textbf{0.844} 		  \\
&	Net4	&	\textbf{0.852} 	&	0.848 		  \\
&	Net5	&	\textbf{0.803} 	&	0.778 		  \\ \hline
\multicolumn{2}{|l|}{{Overall AUPR}} & \textbf{125.345} 	&	103.068 	   	\\ \hline 
\multicolumn{2}{|l|}{{Overall AUROC}} &	\textbf{44.250} 	&	40.962 	   	\\ \hline
\multicolumn{2}{|l|}{Overall Score} &	\textbf{84.798} 	&	71.589 	   	\\ \hline
    \end{tabular}\\
%\end{subtable}
\caption{\label{perf_data}Performance on different data sizes. R is the number of replicates of time-series data. Columns under `R' are the performances without filtration; R* represents performances with filtration and R** represents the performances of the top participants in the DREAM4 Challenge, as shown in Table~\ref{perftop}. The highest scores are marked in bold.}
}
\end{table}

The distribution of the computation time is shown in Fig.2 in the supplementary document. %~\ref{timechart}. 
The linear system was solved by the Matlab backslash ($\backslash$) operator, which is quite efficient and stable. The cross validation has to solve up to 20 times the linear system and to compute the Frobenius norm 20 times to choose a better regularization term $\alpha$ in equation~\ref{P1}. As a consequence the cross validation takes a long chunk of the actual calculation time, but this is not relevant to the time complexity of the core algorithm. Note that the networks and data for these tests were randomly generated, since generating dynamic data for large networks is another challenging task. 

\begin{figure}
\centering
  \includegraphics[width=0.7\linewidth]{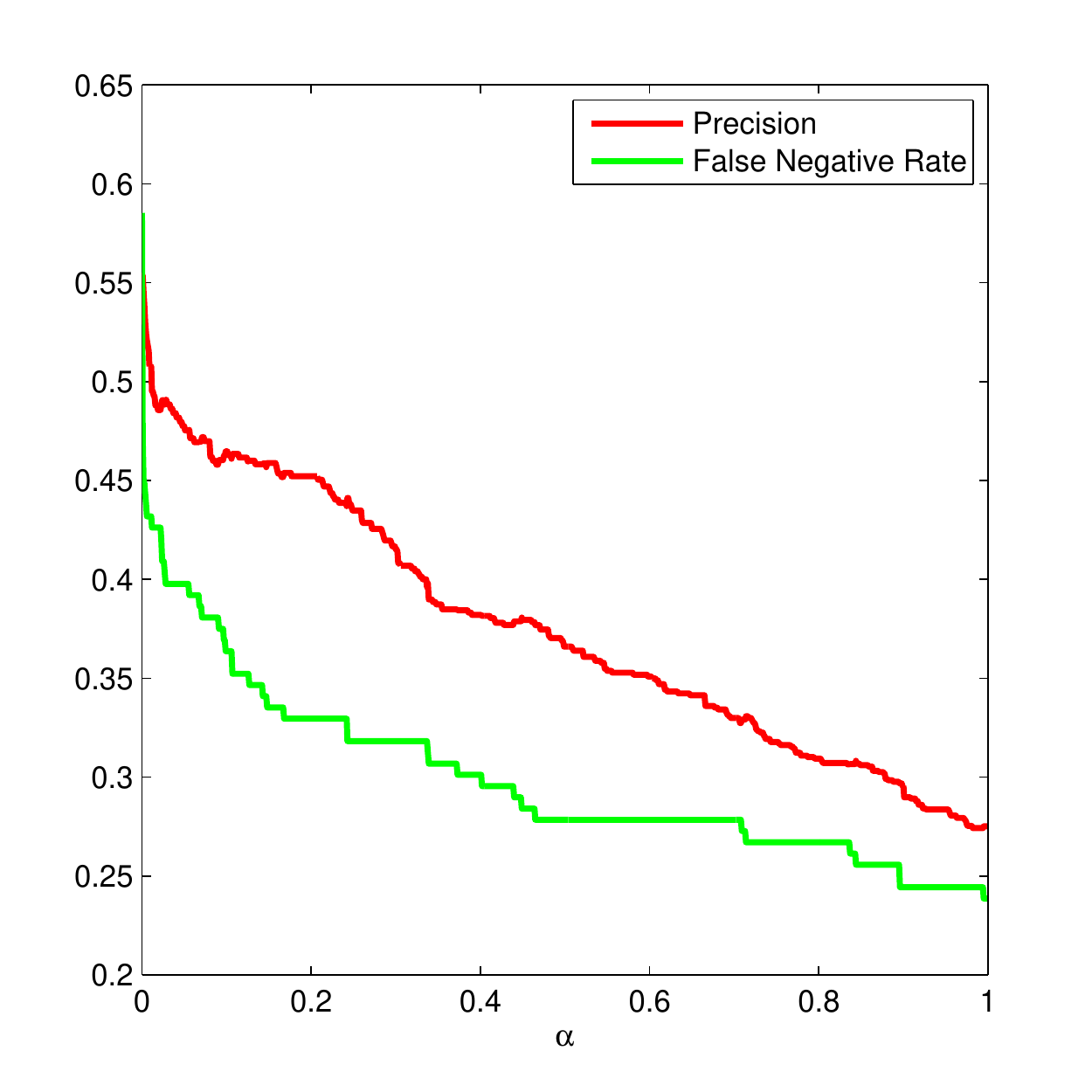}
  \caption{Precision and False Negative Rate (FNR) on network 1 in the DREAM4 challenge vs significance level $\alpha$. The precision decreases as $\alpha$ increases, indicating an increasing \textit{type I error}; a decreasing FNR indicates a decreasing \textit{type II error}.}\label{figalpha}
  \end{figure}
  
\section{Conclusions}

We focused on the identification of parameters for the linear ODE systems. A computationally cheap solution of both unconstrained and constrained optimization problems was devised using the Frobenius norm, so that we could determine parameters in large linear ODE systems from experimental data. Because the experimental data are usually noisy, a series of fitted central difference schemes were provided to handle the derivatives in the ODE system. For sparse ODE systems, we devised a way for outlier detection to introduce sparsity into the system. 

We applied the unconstrained solution to reconstruct gene regulatory networks, the dynamics of which were simplified into a linear ODE system. The results on the 1\,565-gene \textit{E.coli} regulatory networks, of which about 2.5 million parameters have been determined, showed that the proposed method would turn in a satisfactory performance if sufficient data were provided, and that noise-robust difference schemes play an important role in improving the method's performance. In this test, the schemes identified as having better noise to robustness ratios in the theoretical analysis yielded better performances as the noise level increased. 

We showed that the noise analysis can give us prior knowledge that can be used in choosing a finite difference scheme before actually conducting the numerical experiment.
 The constrained solution with filtration to introduce sparsity was evaluated in the  DREAM4 \textit{In Silico 100-gene Network Challenge}. The filtration process demonstrated a powerful ability to reduce the number of parameters in the ODE system from about 10\,000 to about 500. With the data supplied, the performance topped that of all other participants in the DREAM challenge. While filtration is able to reduce the dimension of the required data remarkably, it produces errors and the significance level controls the effectiveness of filtration. A lower significance level filters out more parameters while it has a higher risk of falsely filtering out nonzero parameters. With more data, the unconstrained solution without filtration improved even further the overall scores of the proposed method. This indicates a trade-off that we explored between filtration and data availability. 

Even while demonstrating the advantages of HiDi in tackling the challenge of network reconstruction, we did not find any clear connection between accuracy and the topological properties of the networks being reconstructed. This lack of correspondence between a network feature and the power of our algorithm is compatible with previous findings in~\cite{crowds} for robust gene network inference in the same context of the DREAM Challenge, findings that establish that no method can outperform all others and that the best approach is an heuristic incremental joint effort among the best algorithms. We think that HiDi contributes to this effort.

\section{Acknowledgments}

NAK was supported by a Vinnova VINNMER fellowship, Stratneuro. HZ was supported by the Swedish Research Council. The founders played no role in the design of the study, in data collection and analysis, in the decision to publish, or in the preparation of the manuscript.
%imports the bibliography file "bibliography.bib"

\end{document}